\documentclass[journal]{IEEEtran}
\usepackage[colorlinks,linkcolor=blue,anchorcolor=blue,citecolor=blue,bookmarks=true]{hyperref}  
\usepackage[T1]{fontenc}
\ifCLASSINFOpdf
\usepackage[pdftex]{graphicx} 
\DeclareGraphicsExtensions{.eps,.pdf,.jpeg,.png}
\else
\usepackage[dvips]{graphicx}
\fi
\usepackage[compress,nospace]{cite}
\usepackage[cmex10]{amsmath}
\usepackage{epstopdf,amsthm,stfloats,siunitx,amssymb,wasysym,algorithm,algorithmic,array,url,color,subfigure} 
\usepackage{footnote}
\makesavenoteenv{tabular}
\newtheorem{theorem}{Proposition}

\newcommand{\tabincell}[2]{\begin{tabular}{@{}#1@{}}#2\end{tabular}}
\usepackage{soul}

\begin{document}

\title{Reconfigurable Intelligent Surface Aided Sparse DOA Estimation Method With Non-ULA}

\author{Peng~Chen,~\IEEEmembership{Member,~IEEE},
Zihan~Yang,
Zhimin~Chen,~\IEEEmembership{Member,~IEEE}, 
Ziyu~Guo,~\IEEEmembership{Member,~IEEE}
\thanks{This work was supported in part by the  National Natural Science Foundation of China (Grant No. 61801112), the Natural Science Foundation of Jiangsu Province (Grant No. BK20180357), the National Key R\&D Program of China (Grant No. 2019YFE0120700), and  the Open Program of State Key Laboratory of Millimeter Waves at Southeast University (Grant No. K202029). \textit{(Corresponding author: Peng Chen)}}
\thanks{P.~Chen and Z. Yang are with the State Key Laboratory of Millimeter Waves, Southeast University, Nanjing 210096, China (email: \{chenpengseu, yangzihan\}@seu.edu.cn).}
\thanks{Z.~Chen is with the School of Electronic and Information, Shanghai Dianji University, Shanghai 201306, China (email: chenzm@sdju.edu.cn).}
\thanks{Z.~Guo is with the State Key Laboratory of ASIC and System, Fudan University, Shanghai 201203, China (email: zguo@fudan.edu.cn).}
}

\markboth{IEEE Signal Processing Letters}%
{Shell \MakeLowercase{\textit{et al.}}: Bare Demo of IEEEtran.cls for Journals}

\maketitle

\begin{abstract}
The direction of arrival (DOA) estimation problem is addressed in this paper. A reconfigurable intelligent surface (RIS) aided system for the DOA estimation is proposed. Unlike traditional DOA estimation systems, a low-cost system with only one complete functional receiver is given by changing the phases of the reflected signals at the RIS elements to realize the multiple measurements. Moreover, an atomic norm-based method is proposed for the DOA estimation by exploiting the target sparsity in the spatial domain and solved by a semi-definite programming (SDP) method. Furthermore, the RIS elements can be any geometry array for practical consideration, so a transformation matrix is formulated and different from the conventional SDP method. Simulation results show that the proposed method can estimate the DOA more accurately than the existing methods in the non-uniform linear RIS array.
\end{abstract}

\begin{IEEEkeywords}
	DOA estimation, reconfigurable intelligent surface, atomic norm, sparse reconstruction, non-uniform linear  array.
\end{IEEEkeywords}

\section{Introduction} \label{sec1}
Recently, reconfigurable intelligent surface (RIS)~\cite{9117136,9034159} has been studied widely in the wireless communication, signal processing, and radar fields~\cite{8941126,9133157}. The RIS can reflect the signal and change the amplitude or the phase of the received signal. Based on this characteristic, the RIS can be used to expand signal coverage. Ref.~\cite{9247315} shows that the RIS-assisted systems can improve the wireless communication performance in indoor and outdoor scenarios. For the practical case with only a limited number of discrete phase shifts, a hybrid beamforming scheme is proposed, and the sum-rate maximization problem is formulated in~\cite{9110889}. For the unmanned aerial vehicle (UAV) communication, a joint UAV trajectory and RIS's passive beamforming method are given in~\cite{8959174} to maximize the average achievable data rate.

To obtain the signal direction, the direction of arrival (DOA) estimation methods have been studied for decades~\cite{9146196,7903732,9411879}. The sparse reconstruction-based methods have been proposed to improve the estimation performance~\cite{7314978,9296231,9016105}. The RIS is also used for the DOA estimation. In~\cite{9173575}, SBLNet as a deep network architecture is proposed for the DOA estimation in the autonomous vehicles with RIS. In~\cite{9354904}, a channel estimation problem is considered in a millimeter-wave (mmWave) multiple-input and multiple-output (MIMO) system with RIS, and the channel estimation method is proposed based on a multidimensional DOA estimation method. Then, the corresponding DOA estimation method is compared with the existing methods.

In this paper, for the DOA estimation, a low-cost system is formulated using RIS with only one complete functional receiver. Additionally, different from the traditional uniform linear array (ULA), a non-uniform linear RIS  array (NULRA) is considered, where the RIS elements cannot be at ideal positions for the practical environment.   Then, an atomic norm-based method is proposed for the DOA estimation by exploiting the target sparsity in the space domain and is solved efficiently by a semi-definite programming (SDP) method. Different from the conventional SDP method, a transformation matrix is proposed in the SDP to solve the geometry problem caused by the NULRA. Simulation results are carried out and compared with the existing estimation methods. The main contributions of this paper are given as follows: We consider a system model to estimate the DOA of targets using RIS, where an ideal ULA cannot be formulated on the surface of a practical building. In the existing atomic norm-based methods, since a Toeplitz matrix is formulated, only the ULA can be used for the DOA estimation, limiting the application of atomic norm-based methods. We introduce a transforming matrix and extend the atomic norm-based method to the scenario with NULRA. Based on the transformation matrix, an SDP method is formulated from the atomic norm-based optimization problem, and the details in obtaining the SDP problem are provided.

The remainder of this paper is organized as follows. The NULRA model for DOA estimation is given in Section~\ref{sec2}. Then, the atomic norm-based DOA estimation method is proposed in Section~\ref{sec3}. Simulation results are given in Section~\ref{sec4}, and Section~\ref{sec5} concludes the paper. 

\textit{Notations:} Upper-case and lower-case boldface letters denote matrices and column vectors, respectively.  $(\cdot)^\text{T}$ and $(\cdot)^\text{H}$ are the transpose and the  Hermitian transpose of a matrix, respectively.  $\operatorname{diag}\{\boldsymbol{a}\}$ denotes a diagonal matrix with the diagonal elements from $\boldsymbol{a}$. $\otimes$ denoted the Kronecker product. $\text{Tr}\{\cdot\}$ is the trace of a matrix. $\text{vec}\{\boldsymbol{A}\}$ denotes the vectorization of $\boldsymbol{A}$. $\boldsymbol{A}^\dagger$ is the Moore-Penrose inverse of a matrix $\boldsymbol{A}$, and is defined as $\boldsymbol{A}^\dagger\triangleq (\boldsymbol{A}^{\text{H}}\boldsymbol{A})^{-1}\boldsymbol{A}^{\text{H}}$.

\section{NULRA Model for DOA Estimation}\label{sec2}
\begin{figure}
	\centering
	\includegraphics[width=2in]{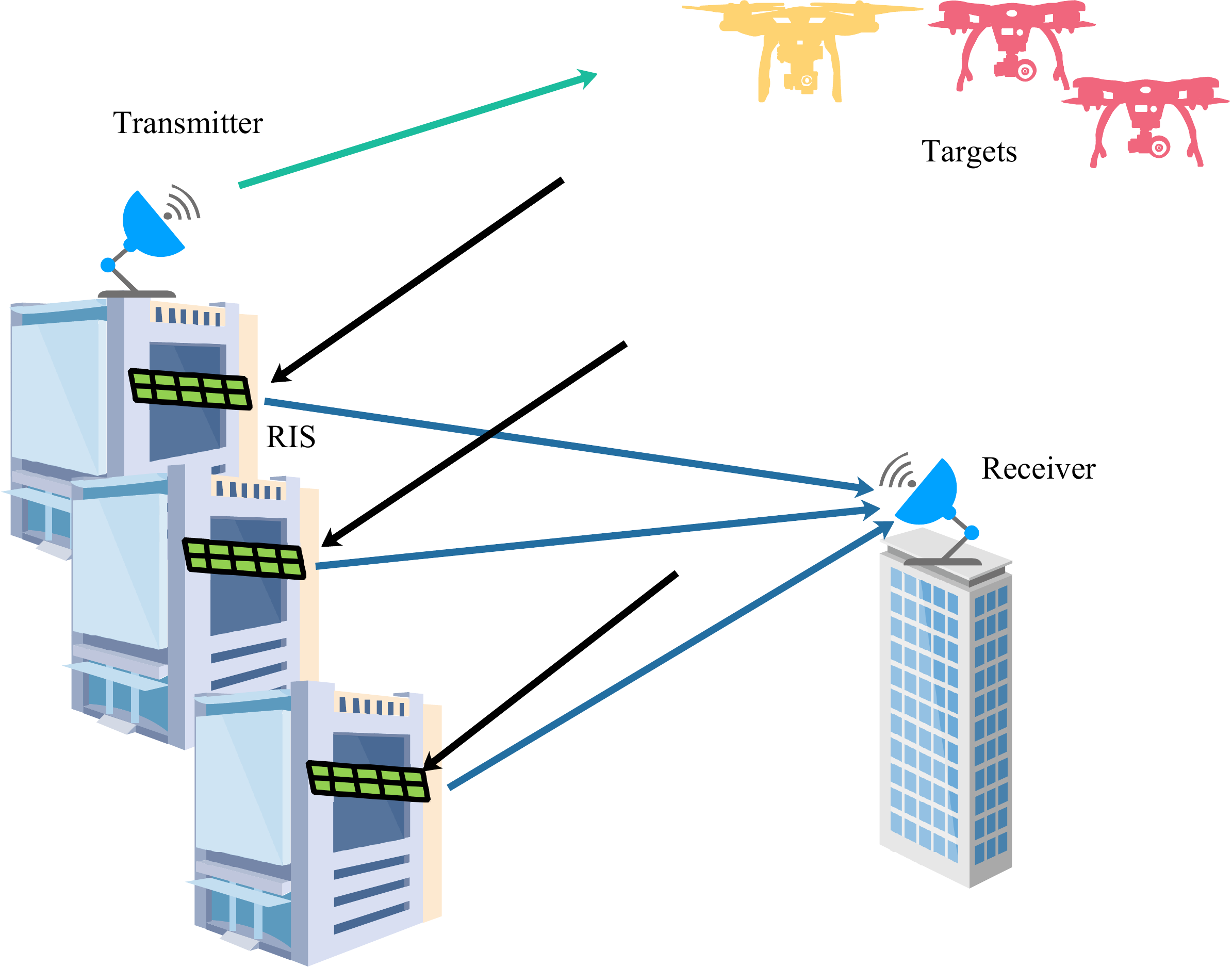}
	\caption{The system model of DOA estimation using NULRA.}
	\label{system}
\end{figure} 

This paper addresses the DOA estimation problem in the system using RIS, reducing the cost significantly via the convenient phase control. With a varactor diode, the phase of the reflected signal from the RIS can be changed easily. Hence, we propose a low-cost system using RIS to estimate the directions of low-speed targets, where the RIS can be placed in a non-uniform linear array, as shown in Fig.~\ref{system}. Additionally, the sparsity of the targets is exploited to improve the DOA estimation performance. 

For a low-cost directing finding system, we use $N$ RIS elements to form an array, where one dimension is considered. The system can be easily extended to the two dimensions of the DOA estimation problem. We try to form a ULA with the distance between the adjacent RIS elements half wavelength. Since the RIS can be placed on any surface, the practical array can differ from the ULA. The position of the $n$-th RIS element is denoted as $p_n$ ($n=0,1,\dots, N-1$), and can be expressed as $
	p_n = \frac{n\lambda}{2}+\tilde{p}_n$, where $\lambda$ denotes the wavelength, and $\tilde{p}_n$ is the difference between the ULA and the practical array.

The reflected signal in the $n$-th RIS element during the $m$-th ($m=0,1,\dots, M-1$) measurement slot can be expressed as 
\begin{align}
	x_{n,m} = A_{n,m}e^{j\phi_{n,m}}\sum^{K-1}_{k=0}s_ke^{j\frac{2\pi p_n}{\lambda}\sin\theta_k},
\end{align} 
where the number of targets is $K$, the direction of the $k$-th target is $\theta_k$, and the signal from the $k$-th target is $s_k$. $A_{n,m}$ and $\phi_{n,m}$ denote the reflection amplitude and phase at the $n$-th RIS element, respectively. The reflected signals are received by a single antenna system at the direction of $\alpha$, and can be written as
\begin{align}
	r_m = \sum^{N-1}_{n=0} x_{n,m} e^{j\frac{2\pi p_n}{\lambda}\sin\alpha}+w_m,
\end{align}
where $w_m$ denotes the additive white Gaussian noise (AWGN).

Collect the received signals into a vector, and we have
\begin{align}
	\boldsymbol{r}&\triangleq 
	\begin{bmatrix}
		r_0, r_1, \dots, r_{M-1}
	\end{bmatrix}^{\text{T}}\label{eq4}\\
	& = \sum^{N-1}_{n=0} e^{j\frac{2\pi p_n}{\lambda}\sin\alpha} \left(\sum^{K-1}_{k=0}s_k e^{j\frac{2\pi p_n}{\lambda}\sin\theta_k}\right)\boldsymbol{b}_n+\boldsymbol{w}\notag\\ 
& = \boldsymbol{B}\text{diag}\{ \boldsymbol{a}(\alpha, \boldsymbol{p}) \} \boldsymbol{A}(\boldsymbol{\theta}, \boldsymbol{p})\boldsymbol{s} +\boldsymbol{w}, \notag	
\end{align}
where we denote $
	\boldsymbol{w} \triangleq\begin{bmatrix}
		w_0,w_1,\dots,w_{M-1}
	\end{bmatrix}^{\text{T}}$, $ 
	\boldsymbol{s} \triangleq\begin{bmatrix}
		s_0,s_1,\dots,s_{K-1}
	\end{bmatrix}^{\text{T}}$, $ 
		\boldsymbol{\theta} \triangleq\begin{bmatrix}
		\theta_0,\theta_1,\dots,\theta_{K-1}
	\end{bmatrix}^{\text{T}}$, $
	\boldsymbol{b}_n \triangleq\begin{bmatrix}
		A_{n,0}e^{j\phi_{n,0}}, A_{n,1}e^{j\phi_{n,1}},\dots, A_{n,M-1}e^{j\phi_{n,M-1}}
	\end{bmatrix}^{\text{T}}$, $ 
	\boldsymbol{B}  \triangleq \begin{bmatrix}
		\boldsymbol{b}_0,\boldsymbol{b}_1,\dots, \boldsymbol{b}_{N-1}
	\end{bmatrix}$. Additionally, we define the steering matrix as $	\boldsymbol{A}(\boldsymbol{\theta}, \boldsymbol{p})\triangleq e^{j\frac{2\pi}{\lambda}\boldsymbol{p}\sin^{\text{T}}\boldsymbol{\theta}}$, and the steering vector is defined as $ 	\boldsymbol{a}(\theta, \boldsymbol{p})\triangleq e^{j\frac{2\pi}{\lambda}\boldsymbol{p}\sin \theta}$. 

Finally, from the system model (\ref{eq4}), we try to estimate the direction $\boldsymbol{\theta}$ from the received signal $\boldsymbol{r}$ in the scenario with NULRA $\boldsymbol{p}$. Moreover, only one receiving channel is used so that the system cost can be reduced significantly. To estimate the directions of multiple targets, the multiple measurements are realized by the measurement matrix $\boldsymbol{B}$. 

\section{Atomic Norm-Based DOA Estimation Method}\label{sec3}
To estimation the DOA $\boldsymbol{\theta}$ from the received signal $\boldsymbol{r}$ using the low-cost RIS system, we propose an atomic norm-based method to exploit the target sparsity in the spatial domain. Unlike traditional methods based on the atomic norm, the ULA is considered to formulate the polynomial, and the DOA is estimated from the peak values of the polynomial. We try to propose a new type of atomic norm-based method for the NULRA and multiple measurements.

First, we define an atomic norm as  $\|\boldsymbol{x}\|_{A}\triangleq \inf\bigg\{\sum_n a_n: \boldsymbol{x}=\sum_n a_ne^{j\phi_n} \boldsymbol{a}(\theta_n, \boldsymbol{p}), 
\phi\in[0,2\pi), a_n\geq 0\bigg\}$. Then, from the system model (\ref{eq4}), the sparse reconstruction problem can be expressed as
\begin{align}
	\min_{\boldsymbol{x}} \frac{1}{2}\|\boldsymbol{r}-\boldsymbol{B}\text{diag}\{\boldsymbol{a}(\alpha, \boldsymbol{p})\}\boldsymbol{x}\|^2_2+\gamma \|\boldsymbol{x}\|_A,\label{eq13}
\end{align}
where $\gamma$ is used to achieve a tradeoff between the reconstruction accuracy and the sparsity. To solve the sparse reconstruction problem (\ref{eq13}), we have the following proposition.
\begin{theorem}\label{pro1}
The sparse reconstruction problem (\ref{eq13}) with the atomic norm can be expressed as the following SDP problem
\begin{align}
\min_{\boldsymbol{q},\boldsymbol{G}}& \quad  \left \|\boldsymbol{r}-(\boldsymbol{T}^{\text{H}}\boldsymbol{C}^{\text{H}})^{\dagger}\boldsymbol{q} \right\|^2_2   \notag\\
\text{s.t.}& \quad \begin{bmatrix}\boldsymbol{G} & \boldsymbol{q}\\ \boldsymbol{q}^{\text{H}} & u\end{bmatrix}\succeq 0\text{, } \text{Tr}\{\boldsymbol{G}\}=\gamma^2/u
\label{eq14}\\
&\quad \boldsymbol{G}\in\mathbb{C}^{N\times N} \text{ is a Hermitian matrix}\notag\\
&\quad \sum_n G_{n, n+n'} = 0,\quad n'=1-N,\dots,N-1.\notag
\end{align}
\end{theorem}
\begin{proof}
The sparse reconstruction problem can be rewritten as
\begin{align}
	\min_{\boldsymbol{x},\boldsymbol{y}}  \frac{1}{2}\|\boldsymbol{r}-\boldsymbol{y}\|^2_2+\gamma \|\boldsymbol{x}\|_A,\qquad  
	\text{s.t. } \boldsymbol{y}=\boldsymbol{C}\boldsymbol{x},\notag
\end{align}
where we define $\boldsymbol{C}\triangleq \boldsymbol{B}\text{diag}\{\boldsymbol{a}(\alpha, \boldsymbol{p})\}$.
We can obtain a Lagrangian function as
\begin{align}
	f(\boldsymbol{x},\boldsymbol{y}, \boldsymbol{z})= &\frac{1}{2}\|\boldsymbol{r}-\boldsymbol{y}\|^2_2+\gamma \|\boldsymbol{x}\|_A  +\langle\boldsymbol{z}, \boldsymbol{y}-\boldsymbol{C}\boldsymbol{x} \rangle.
\end{align}
Then, the dual problem of the optimization problem (\ref{eq13}) can be expressed as
\begin{align}
\max_{\boldsymbol{z}}\min_{\boldsymbol{x}, \boldsymbol{y}} f(\boldsymbol{x}, \boldsymbol{y}, \boldsymbol{z}),\label{eq17}
\end{align}
where we have $
\min_{\boldsymbol{x}, \boldsymbol{y}} f(\boldsymbol{x}, \boldsymbol{y}, \boldsymbol{z}) = \min_{ \boldsymbol{y}}\frac{1}{2}\|\boldsymbol{r}-\boldsymbol{y}\|^2_2 +\langle\boldsymbol{z}, \boldsymbol{y} \rangle+ \min_{\boldsymbol{x}} \gamma \|\boldsymbol{x}\|_A  -\langle\boldsymbol{z}, \boldsymbol{C}\boldsymbol{x} \rangle$. 
Take the gradient of the first part to $0$, and we have the optimal value of $\boldsymbol{y}$ is $(\boldsymbol{r}-\boldsymbol{z})$.  Substitute it into (\ref{eq17}), and we have
\begin{align}
\max_{\boldsymbol{z}}\bigg\{\frac{1}{2}\big(\|\boldsymbol{r}\|^2_2-\|\boldsymbol{r}-\boldsymbol{z}\|^2_2 \big) -\max_{\boldsymbol{x}}\big\{ \langle\boldsymbol{z}, \boldsymbol{C}\boldsymbol{x} \rangle - \gamma \|\boldsymbol{x}\|_A\big\} \bigg\}.
\end{align}
 
Since the dual norm of the atomic norm is defined as $\|\boldsymbol{x}\|_{\tilde{A}}\triangleq \sup_{\|\boldsymbol{g}\|_A\leq 1}\langle \boldsymbol{x}, \boldsymbol{g}\rangle$, we have
\begin{align}
\min_{\boldsymbol{z}}  \frac{1}{2} \|\boldsymbol{r}-\boldsymbol{z}\|^2_2, \qquad 
\text{s.t.} \|\boldsymbol{C}^{\text{H}}\boldsymbol{z}\|_{\tilde{A}}\leq \gamma.
\end{align} 
 
The constrain can be simplified as
\begin{align}
\|\boldsymbol{C}^{\text{H}}\boldsymbol{z}\|_{\tilde{A}} & = \sup_{\|\boldsymbol{g}\|_A\leq 1}\langle \boldsymbol{C}^{\text{H}}\boldsymbol{z}, \boldsymbol{g}\rangle\notag\\
& = \sup_{\substack{\sum_n a_n\leq 1, \phi\in[0,2\pi),\\
		 a_n\geq 0,\theta_n\in(-\frac{\pi}{2},\frac{\pi}{2}]}} \sum_n\mathcal{R}\left\{ a_ne^{-j\phi_n} \boldsymbol{a}^{\text{H}}(\theta_n, \boldsymbol{p})\boldsymbol{C}^{\text{H}}\boldsymbol{z}\right\}\notag\\
& =  \sup_{\substack{\sum_n a_n\leq 1, a_n\geq 0,\\
		\theta_n\in(-\frac{\pi}{2},\frac{\pi}{2}]}} \sum_n a_n| \boldsymbol{a}^{\text{H}}(\theta_n, \boldsymbol{p})\boldsymbol{C}^{\text{H}}\boldsymbol{z}|\notag\\
& =  \sup_{\theta\in(-\pi/2,\pi/2]}| \boldsymbol{a}^{\text{H}}(\theta, \boldsymbol{p})\boldsymbol{C}^{\text{H}}\boldsymbol{z}|.\label{eq22}
\end{align} 
For a Hermitian matrix $\begin{bmatrix}\boldsymbol{G} & \boldsymbol{q}\\ \boldsymbol{q}^{\text{H}} & u\end{bmatrix}$, it is a  positive semi-definite matrix, if and only if we have $
\boldsymbol{G}\succeq 0$ and $
u\boldsymbol{G}- \boldsymbol{qq}^{\text{H}}\succeq 0$. 
Hence, for any steering vector $\boldsymbol{a}(\theta, \bar{\boldsymbol{p}})$, we have
\begin{align}
|\boldsymbol{a}^{\text{H}}(\theta, \bar{\boldsymbol{p}})\boldsymbol{q}|^2\leq u\boldsymbol{a}^{\text{H}}(\theta, \bar{\boldsymbol{p}})\boldsymbol{G}\boldsymbol{a}(\theta, \bar{\boldsymbol{p}}),\label{eq25}
\end{align}
where the $n$-th entry of $\bar{\boldsymbol{p}}$ is $\bar{q}_n = n\lambda/2$.

With (\ref{eq22}) and (\ref{eq25}), if we set $\boldsymbol{q}=\boldsymbol{T}^{\text{H}}\boldsymbol{C}^{\text{H}}\boldsymbol{z}$ and $u\boldsymbol{a}^{\text{H}}(\theta, \bar{\boldsymbol{p}})\boldsymbol{G}\boldsymbol{a}(\theta, \bar{\boldsymbol{p}})\leq \gamma^2$, we can obtain $
|\boldsymbol{a}^{\text{H}}(\theta, \bar{\boldsymbol{p}})\boldsymbol{T}^{\text{H}}\boldsymbol{C}^{\text{H}}\boldsymbol{z}|^2\leq \gamma^2$, where a transformation matrix $\boldsymbol{T}$ is defined and satisfies $\boldsymbol{Ta}(\theta, \bar{\boldsymbol{p}})=\boldsymbol{a}(\theta, \boldsymbol{p})$. Therefore, the constraint is satisfied.

As a conclusion, the optimization can be rewritten as the following SDP problem
\begin{align}
\min_{\boldsymbol{z},\boldsymbol{G}}& \quad \frac{1}{2} \|\boldsymbol{r}-\boldsymbol{z}\|^2_2   \notag\\
\text{s.t.}& \quad \begin{bmatrix}\boldsymbol{G} & \boldsymbol{q}\\ \boldsymbol{q}^{\text{H}} & u\end{bmatrix}\succeq 0\\
&\quad \boldsymbol{G}\in\mathbb{C}^{N\times N} \text{ is a Hermitian matrix}\notag\\
&\quad \boldsymbol{q}=\boldsymbol{T}^{\text{H}}\boldsymbol{C}^{\text{H}}\boldsymbol{z}\text{, }\text{Tr}\{\boldsymbol{G}\}=\gamma^2/u\notag\\
&\quad \sum_n G_{n, n+n'} = 0,\quad n'=1-N,\dots,N-1,\notag
\end{align}
which can be rewritten as (\ref{eq14}) easily.
\end{proof}

With Proposition~\ref{pro1}, the atomic norm-based DOA estimation method is obtained in the scenario with NULRA. For the transformation matrix $\boldsymbol{T}$, we have $
\boldsymbol{T}\boldsymbol{A}(\boldsymbol{\theta},\bar{\boldsymbol{p}})=\boldsymbol{A}(\boldsymbol{\theta},\boldsymbol{p})$, and can be estimated as 
\begin{align}
\text{vec}\{\boldsymbol{T}\} = (\boldsymbol{A}^{\text{T}}(\boldsymbol{\theta},\bar{\boldsymbol{p}})\otimes \boldsymbol{I}_N)^{\dagger} \text{vec}\{\boldsymbol{A}(\boldsymbol{\theta},\boldsymbol{p})\},
\end{align}
where we have $
\text{vec}\{\boldsymbol{T}\boldsymbol{A}(\boldsymbol{\theta},\bar{\boldsymbol{p}})\}=(\boldsymbol{A}^{\text{T}}(\boldsymbol{\theta},\bar{\boldsymbol{p}})\otimes \boldsymbol{I}_N)\text{vec}\{\boldsymbol{T}\}$.

Using the CVX toolbox, the results of Proposition~\ref{pro1} can be obtained, and we denote the optimal value of $\boldsymbol{q}$ as $\hat{\boldsymbol{q}}$. Then, the DOA can be estimated from the peak values of the following polynomial  $
	f(\theta) = |\boldsymbol{a}^{\text{H}}(\theta, \bar{\boldsymbol{p}})\hat{\boldsymbol{q}}|^2$.

\section{Simulation Results}\label{sec4}
In this section, the simulation results of the DOA estimation method for the NULRA will be given, and the simulations are carried out in a personal computer with a 2.9 GHz Intel Core i5 processor and 8 GB 2133 MHz DDR3. Moreover, the Matlab code for the proposed method is available online (\url{https://github.com/chenpengseu/NonULA-ANM.git}). The simulation parameters are given in Table~\ref{table1}. The mean distance between the adjacent RIS elements is the half wavelength, with the standard derivation being $0.1$ (normalized by wavelength). Additionally, for the RIS, only $2$ phases (i.e., $\ang{0}$ and $\ang{180}$) can be controlled by the system in the measurement matrix. Multiple measurements instead of receiving channels are used to estimate the DOA in this paper, better estimation performance can be achieved by more measurement. As shown in (\ref{eq4}), the number of targets is $K$, and the number of measurements is $M$. For the sparse reconstruction, when the number of measurements is $\mathcal{O}(K\log K\log N)$, it is sufficient to guarantee exact DOA estimation with high probability~\cite{6576276}. 

Note that the switch time $T_s$ of a varactor is serval microseconds ($\mu$s), so with $M$ measurements, the target positions cannot change significantly. With the parameters in Table~\ref{table1}, we can find that the DOA resolution is greater than $\ang{0.1}$, so during the estimation measurement, the limit speed of the target is about $\frac{\ang{0.1}}{\ang{180}}\pi R/(MT_s)$, where the DOA changes are less than $\ang{0.1}$, the switch time $T_s$ is about $10$ $\mu$s and $R$ denotes the range between RIS and targets. Hence, we can find that the speed limit is about $5.4R$ m/s, which is suitable for most applications.

\begin{table}
	\renewcommand{\arraystretch}{1}
	\caption{Simulation parameters}
	\label{table1}
	\centering
	\begin{tabular}{cc}
		\hline
		  \textbf{Parameter} &  \textbf{Value} \\
		\hline 
		The number of RIS elements & $N=16$ \\
		The number of measurements & $M=32$ \\ 
		\tabincell{c}{The mean distance between\\ adjacent RIS elements} & $d=\lambda/2$\\
		The number of targets & $K=3$\\ 
		The directions of targets & $\boldsymbol{\theta} =[-\ang{30.345},\ang{0.789}, \ang{20.456}]$\\
		\tabincell{c}{The standard derivation of RIS position \\ (normalized by the wavelength)} & $0.1$\\
		\hline
	\end{tabular}
\end{table}

\begin{figure}
	\centering
	\includegraphics[width=2.5in]{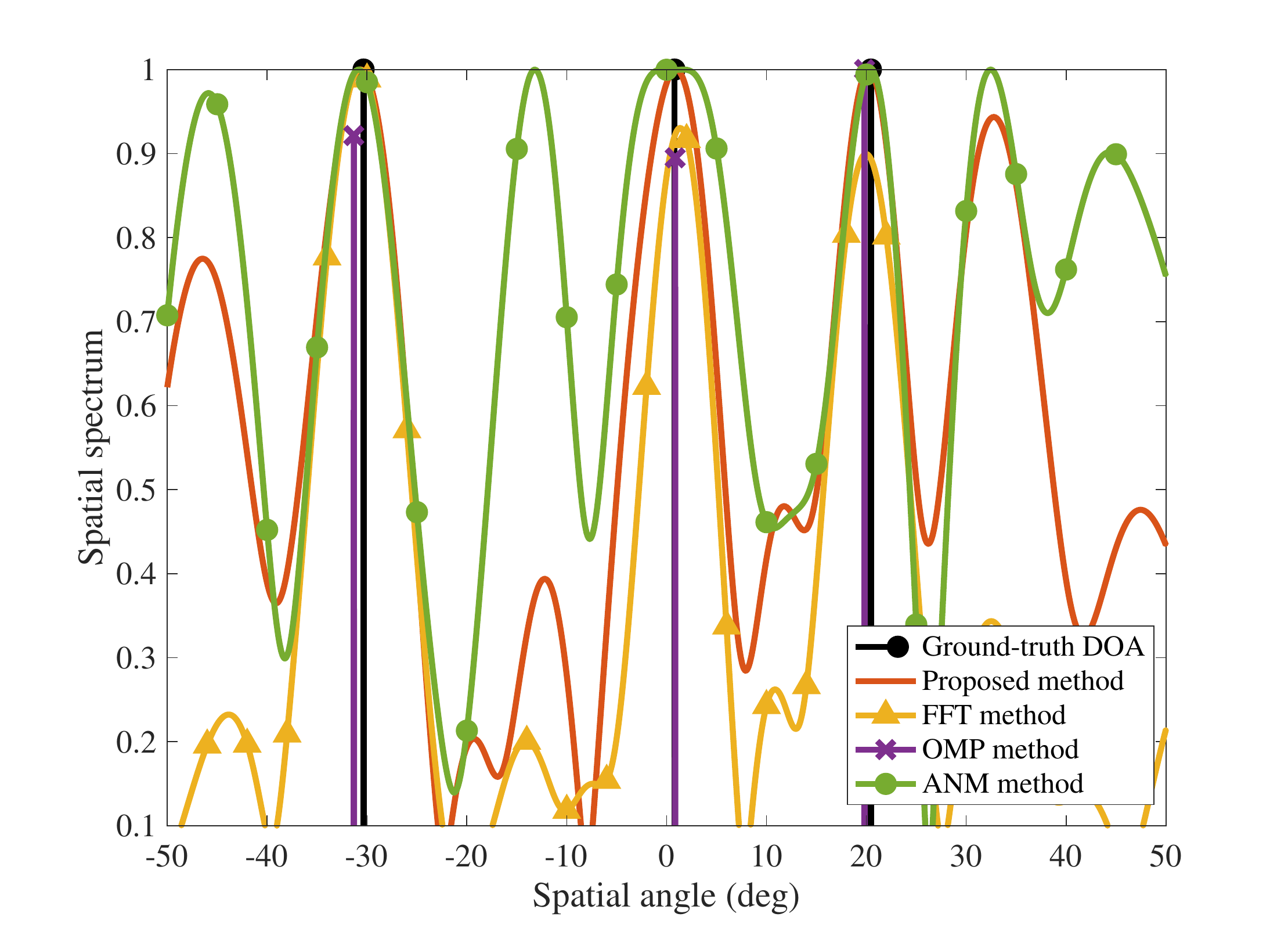}
	\caption{The spatial spectrum for DOA estimation.}
	\label{sp}
\end{figure} 

First, the spatial spectrum for the DOA estimation is given in Fig.~\ref{sp}. The proposed method is compared with the $3$ existing methods, i.e., the fast Fourier transformation (FFT)-based method, the conventional atomic norm minimization (ANM) method, and the orthogonal matching pursuit (OMP) method. As shown in this figure, the root-mean-square error (RMSE) of DOA estimation can be obtained. The RMSEs of the proposed method, the FFT-based method, the OMP method, and the ANM method are $0.218$, $0.482$, $0.673$, and $0.594$ in degree, respectively. The estimation error of the proposed method is much lower than existing methods, which shows the efficiency of the proposed method.

\begin{figure}
	\centering
	\includegraphics[width=2.2in]{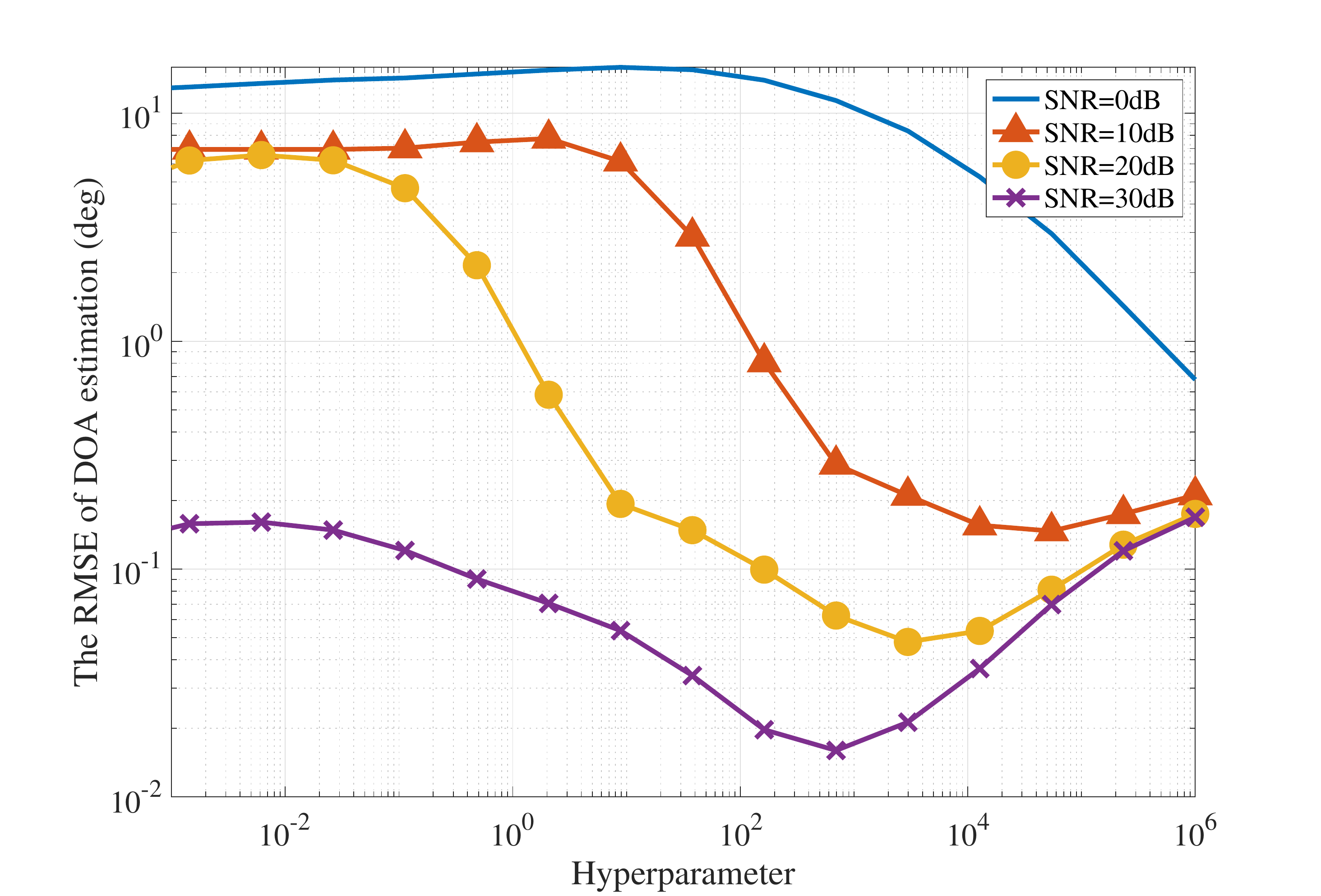}
	\caption{The RMSE of the DOA estimation with different hyperparameters.}
	\label{hyper}
\end{figure} 

Then, the DOA estimation performance with different hyperparameters is also shown in Fig.~\ref{hyper}. As shown in Proposition~\ref{pro1}, the hyperparameters $u$ and $\gamma$ are used. For the simplification, we set $u=1$ and change the hyperparameter $\gamma$. 
As shown in~\cite{7833233}, the value of hyperparameter can be set as $\gamma^2 \simeq 10^{0.1\zeta}MK\log MK$, where $\zeta$ is the SNR in dB. To obtain the exact hyperparameter, the simulation results are shown in Fig.~\ref{hyper} with the SNR of the received signal being $0$ dB, $10$ dB, $20$ dB, and $30$ dB. From the simulation results, we can find better estimation performance by choosing the appropriate value of the hyperparameter. As experience value, the hyperparameter can be chosen as $ 
	\gamma^2=10^{-0.096 \zeta+5.5722}.$

\begin{figure}
	\centering
	\includegraphics[width=2.3in]{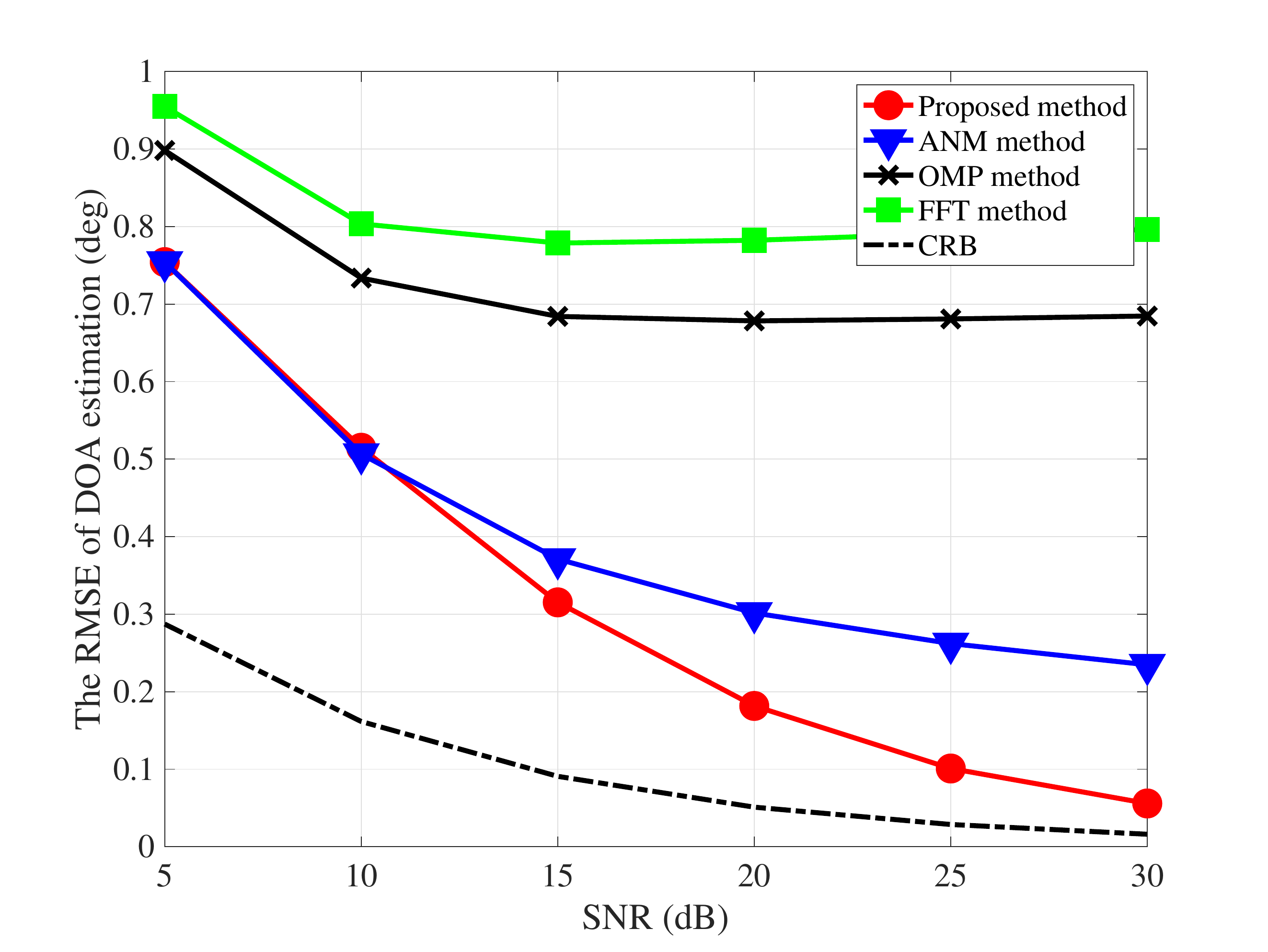}
	\caption{The RMSE of the DOA estimation with different SNRs.}
	\label{SNR}
\end{figure}

\begin{figure}
	\centering
	\includegraphics[width=2.3in]{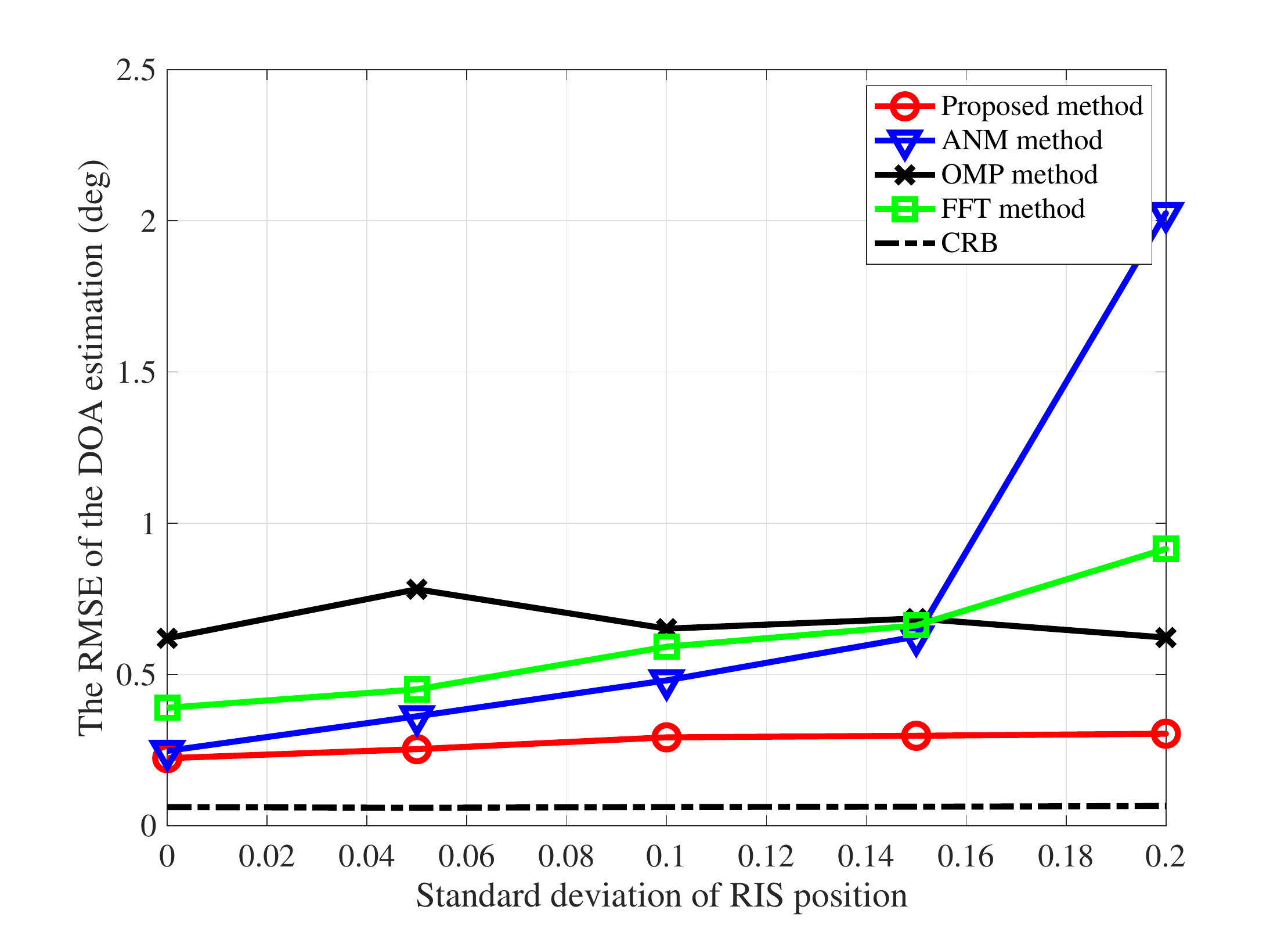}
	\caption{The RMSE of the DOA estimation with different standard deviations of RIS position.}
	\label{Per}
\end{figure} 

\begin{figure}
	\centering
	\includegraphics[width=2.1in]{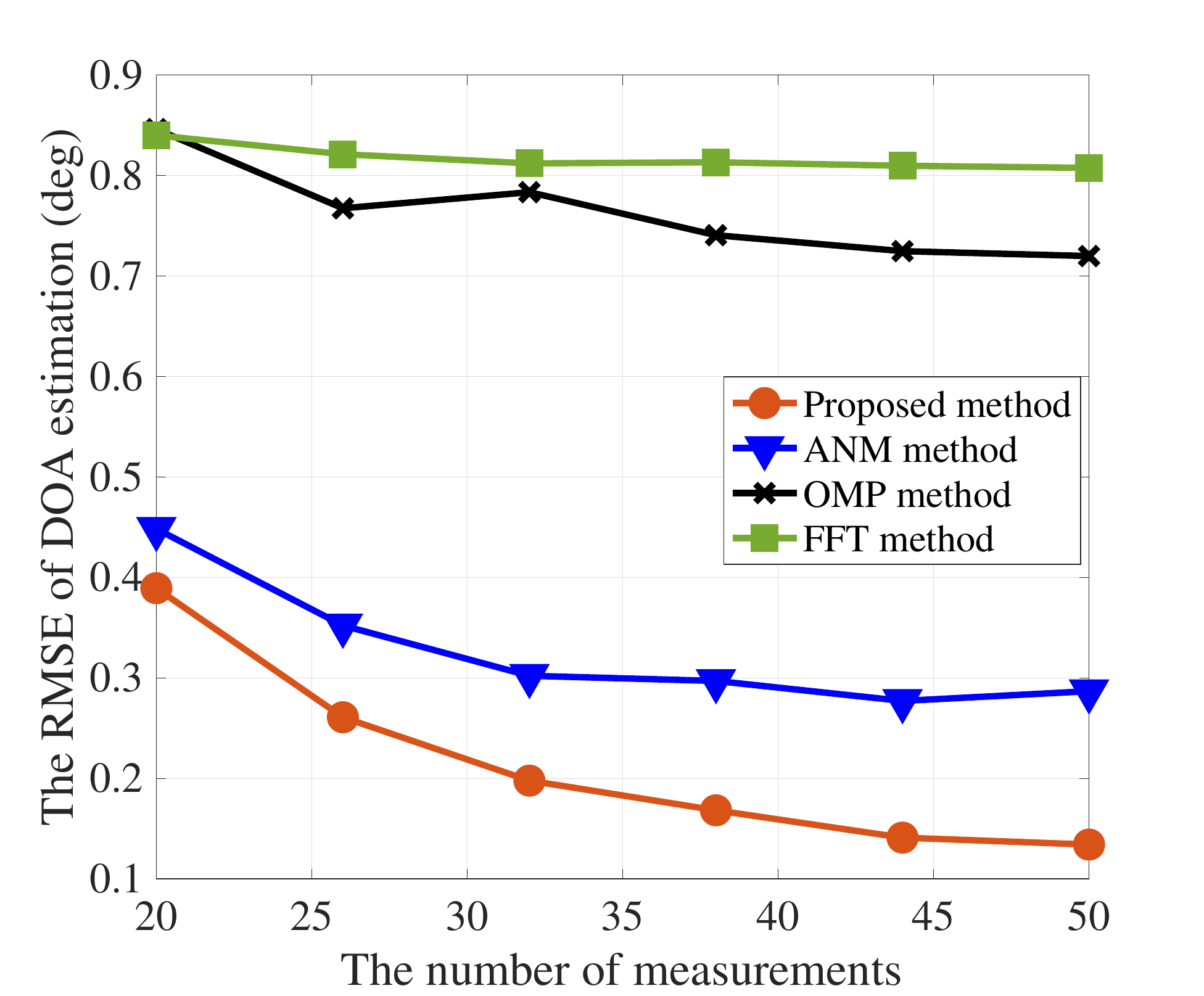}
	\caption{The RMSE of the DOA estimation with different numbers of measurements.}
	\label{mea}
\end{figure} 

By choosing the appropriate value of the hyperparameter, the DOA estimation performance with different SNRs is shown in Fig.~\ref{SNR}. Much better performance is achieved by the proposed method than the existing methods, especially with the SNR being greater than $15$ dB. Moreover, the estimation performance with different standard deviations of RIS position is also shown in Fig.~\ref{Per}. The better estimation performance is achieved by the proposed method. Furthermore, the estimation performance is not affected by the RIS position, so the proposed method is suitable for the NULRA. Moreover, the DOA estimation performance with different numbers of measurements is shown in Fig.~\ref{mea}. We can find that the DOA estimation performance is improved using more measurements, and the better estimation performance is achieved by the proposed method. The computational complexity of the proposed method is shown by the computational time. With the simulation parameters in Table~\ref{table1}, the computational time of the proposed method is $2.25$ s, that of the ANM method is $0.12$ s, that of the OMP method is $0.08$ s, and that of the FFT method is $2.31$ s. Hence, the computational complexity of the proposed method is also the same as the ANM method, which is much higher than OMP and FFT methods.

\section{Conclusions}\label{sec5}
 The DOA estimation method has been studied in the scenario with the NULRA. To exploit the target sparsity in the spatial domain, the atomic norm-based method has been proposed, where a new type of atomic norm is defined with the geometry structure of the RIS array. The SDP problem has been formulated with the transformation matrix to solve the atomic norm-based method, which is different from the existing methods. Simulation results have shown the efficiency of the proposed method in the DOA estimation problem using the NULRA. Future work will focus on the DOA estimation method with lower computational complexity and the jointly designed measurement matrix to optimize the receiving SNR and the performance of the sparse reconstruction.

\bibliographystyle{IEEEtran}
\bibliography{References.bib}

\end{document}